\documentclass[conference, 10pt, twocolumn]{ieeeconf}

\IEEEoverridecommandlockouts
\usepackage[usenames]{color}
\usepackage[colorlinks=false,urlcolor=blue,citecolor=blue,linkcolor=blue,
bookmarks=true,
     bookmarksopen=false,pdftitle=created by dvipdf,pdfcreator=NM,
     pdfauthor=Megiddo,pdfsubject=mor.sty:6.29.04]{hyperref}
\usepackage{url}
\usepackage{subfigure}
\usepackage{amsfonts,mathrsfs}
\usepackage{amssymb,amsmath}
\usepackage{verbatim}
\usepackage{acronym}
\usepackage{graphicx}
\usepackage{enumerate}

\def\fskip#1{}

\newtheorem{theorem}{Theorem}

\newtheorem{definition}{Definition}

\newtheorem{lemma}{Lemma}

\renewcommand{\theenumi}{\arabic{enumi}}

\def\1{{\bf 1}}

\def\e{{\bf e}}

\def\F{\mathscr{F}}

\newcommand{\lra}{\leftrightarrow}

\newcommand{\Wj}{W^{(j)}}

\def\e{\epsilon}

\newcommand{\al}{\alpha}
\def\iot{i_1(t)}
\newcommand{\itt}{i_2(t)}

\def\R{\mathbb{R}}

\newcommand{\EXP}[1]{\mathsf{E}\!\left[#1\right] }
\newcommand{\prob}[1]{\mathsf{Pr}\left( #1 \right)}

\newcommand{\bfo}{\mathbf{1}}

\newcommand{\Ej}{E^{(j)}}

\newcommand{\Ft}{\{\mathcal{F}(t)\}}

\newcommand{\Gj}{G^{(j)}}

\newcommand{\Rm}{\mathbb{R}^m}

\newcommand{\Xj}{X^{(j)}}

\newcommand{\todo}[1]{\vspace{5 mm}\par \noindent \marginpar{\textsc{ToDo}}
\framebox{\begin{minipage}[c]{0.9 \columnwidth} \tt #1
\end{minipage}}\vspace{5 mm}\par}

\begin{document}
\title{A General Framework for Distributed Vote Aggregation}
\author{\authorblockN{Behrouz Touri, Farzad Fardnoud, Angelia Neidi\'c, and Olgica Milenkovic}
  \authorblockA{Coordinated Science Laboratory, University of Illinois at Urbana-Champaign,  Urbana, IL 61801\\
     Email: (touri1,hassanz1,milenkov,angelia)@illinois.edu}}
\maketitle
\begin{abstract}
We present a general model for opinion dynamics in a social network
together with several possibilities for object selections at times when the agents are communicating.
We study the limiting behavior of such a dynamics and show that this dynamics almost surely converges. We consider some special implications of the convergence result for gossip and top-$k$ selective gossip
models. In particular, we provide an answer to the open problem of the convergence property of the top-$k$ selective gossip model, and show that the convergence holds in a much more general setting. Moreover, we propose an extension of the gossip and top-$k$ selective gossip models and provide some results for their limiting behavior.
\end{abstract}


\section{Introduction}\label{sec:introduction}
Voting systems have recently gained significant attention due to the emergence of complex online marketing industries, novel forms of political voting systems and opinion pole dynamic aggregation, as well as
many new applications that reach well beyond social choice theory and computer science, such as signal processing, bioinformatics, coding theory, and machine learning~\cite{Mueller1976}, \cite{dwork2001rank-web}, \cite{Cook:1978vn}.
In connection with social network modeling and analysis, one challenging area in the context of voting systems is to model and describe how the opinion of agents about various candidates varies with time,
as agents in a society interact with each other and are exposed to opinion influence by media and their environment.

We describe a novel and broad general model for opinion dynamics analysis of voting systems within a dynamically changing (evolving) society. We investigate the limiting behavior of this dynamics and its relation with the connectivity pattern of the society. The proposed dynamics is based on the gossip
averaging dynamics of~\cite{Boyd06}, but it includes more general notions of ``random connectivity'' in a society, for example,
the possibility of discussing voters opinions. Thus, the application domain of the proposed model is not restricted to opinion dynamics analysis only and, it has connections to several interesting distributed algorithms over time-varying networks. In particular, this work is closely related to the works in the area of distributed averaging and consensus literature \cite{Tsitsiklis84}, \cite{Jadbabaie03}, \cite{Boyd06}, \cite{RabbatCDC:2012} and opinion dynamics \cite{Krause:02}, \cite{Lorenz}. A recent addition to the ever-growing consensus literature also includes a spectral graph theory approach for voting over networks, described in ~\cite{chung2012hypergraph}. There, voter preferences are expressed via vertex colors, and hyperedges are used to denote interacting group of agents or agents subjected to the same kind of external influence. The topology of the hypergraph is fixed.

The structure of this paper is as follows: In Section~\ref{sec:generaldynamics}, we introduce our novel notion of network dynamics,
which we refer to it as \textit{voting diffusion dynamics}, and we discuss some instances of such dynamics which are both of practical and theoretical interest.
Then, in Section~\ref{sec:convergence}, we prove the stability of this dynamics in a general setting and provide a characterization of the limiting point of the dynamics.
In Section~\ref{sec:implications}, we present and study the implications of our result to generalizations of the classical gossip algorithm in~\cite{Boyd06} and the top-$k$ selective gossip algorithm~\cite{RabbatCDC:2012}.
We provide the answer to an open problem regarding the convergence property of top-$k$ selective algorithm that was introduced in~\cite{RabbatCDC:2012}. We in fact provide an answer to the posed question
in a much more general settings of random social connectivity. Conclusions are given in Section~\ref{sec:conclusion}.

 \textbf{Notation}:
For any $n\geq 1$, we denote the set $\{1,\ldots,n\}$ by $[n]$. We say that an $m\times m$ matrix $A$ is stochastic if $\sum_{j=1}^{m}A_{ij}=1$ for all $i\in [n]$ and for all $i,j\in [m]$, we have $A_{ij}\geq 0$.
We use $e$ to denote a vector with all entries equal to 1, where the size of the vector is understood from the context. Throughout the paper, we let $(\Omega,\mathcal{F},\prob{\cdot})$ be the underlying probability space, where $\Omega$ denotes the sample space, $\mathcal{F}$ denotes the $\sigma$-algebra of measurable sets, while $\prob{\cdot}$ stands for the underlying probability measure. We say that $\Ft$ is a filtering if $\mathcal{F}_t$ is a $\sigma$-algebra of $\Omega$ for all $t\geq 0$ and $\mathcal{F}_0\subseteq\mathcal{F}_1\subseteq \cdots$. We denote the conditional expectation by $\EXP{\cdot\mid \cdot}$, and the indicator function
of an event $E$ by $\bfo_E$, i.e.\ $\bfo_E(\omega)=1$ if $\omega\in E$ and $\bfo_E(\omega)=0$ otherwise. We refer to a collection of random variables $W_{ij}$ indexed by $i,j\in [m]$ as a random $m\times m$ matrix.

\section{Dynamic System Viewpoint}\label{sec:generaldynamics}
In this section, we discuss the general setting for our dynamics and introduce the voting diffusion model.
We consider a society of $[m]$ agents that are connected through social ties. Each individual in this society has an opinion about $n$ objects or services of the same kind such as, movies, books, political parties, or services of different dentists in a city. We refer to those objects that are to be ranked as \textit{candidates}. We assume that each individual scores each of the candidates with a real number. Thus, each individual's opinion about the $n$ candidates can be represented by a vector in $\R^n$: the larger the entry, the better the opinion of the agent is about the corresponding candidate. We encode the initial belief (score) of the $m$ individuals
about the $n$ candidates by an $m\times n$ matrix $X(0)$. Thus, $X_{ij}(0)$ represents the opinion of the $i$th agent about the $j$th candidate. We refer to the vector that represents the opinion of the $i$th agent on
all the available candidates, i.e.,\ the $i$th row of $X(0)$, as the \textit{opinion profile} of the $i$th person.

We also refer to $X(0)$ as the opinion profile of the society at time $0$, or simply as the initial opinion profile of the society.

Our interest is in the dynamics of distributed rank influence and aggregation in a society and how one can model diffusion of opinions in a society about different candidates. Here, aggregation refers to the process of assembling individual votes into one vote representative of the whole social opinion~\cite{Mueller1976}. Aggregation is usually performed via scoring methods, plurality counts or using specialized distance measures~\cite{Mueller1976}. As for the aggregation method, there is no common consensus on the best way of combining $m$ votes. Here, we assume that ranking is performed through scoring and hence, we focus on weighted Borda method, for rank aggregation. In Borda's approach,
the aggregate of the $m$ rankings is simply the ordering of the candidates based on their average score within the society. Specifically,
the average score is given by $\bar{X}(0)=\frac{1}{m}e^TX(0)$, and the aggregate of the rankings is
given by the ordering of the average scores, i.e.,\ a permutation $\sigma:[n]\to[n]$ such that $\bar{X}_{\sigma_1}(0)\geq \bar{X}_{\sigma_2}(0)\geq \cdots \ge\bar{X}_{\sigma_n}(0)$.
In other words, $\sigma_1$ is the index of the candidate with the highest aggregate score (vote). We refer to the vector $\bar{X}(0)$ as the \emph{aggregate profile}.

As discussed in \cite{RabbatCDC:2012}, in many occasions, when agents exchange their opinion in a social setting, they would not necessarily exchange their beliefs about every single candidate. As an example, when people in a group discuss movies, even if they recall all the movies they watched, they would not necessary talk about every single one of them. However, even from time to time, they may exchange their beliefs about different number of candidates. For example, one may discuss two movies with a friend at some time, or four movies with some other friend at some other time.

We study the opinion dynamics of $m$ agents about $n$ candidates using a general form of gossiping to model the belief exchange among the agents.
The opinion dynamics in the society is represented by a (random) matrix process $\{X(t)\}$,
where $X_{ij}(t)$ is the opinion of the $i$th agent about the $j$th object.

\subsection{Voting Diffusion Model}
The general dynamics that we analyze has the following properties:
  \begin{enumerate}[(I)]
   \item The dynamics starts with some arbitrary (random) opinion profile $X(0)$.
   \item The dynamics evolves in discrete time, where the time is indexed
   by non-negative integers $t\geq 0$.
   \item We have a random sequence $\{\{i_1(t),i_2(t)\}\}$ of pairs of agents
   (i.e.,\ $i_1(t),i_2(t)\in [m]$ almost surely). We refer to $\{\{\iot,\itt\}\}$ as the \textit{communicating pair process}. Here, we assume that $\{i_1(t),i_2(t)\}$ is measurable with respect to discrete measure on
   the set $\{\{i,j\}\mid i\not=j, i,j\in [m]\}$.

   \item Agents $i_1(t),i_2(t)$ discuss and update their opinion about items in a random (measurable) set $S(t)\subseteq[n]$ of candidates, as follows:
       \begin{align}\label{eqn:dynrule}
         X_{ij}(t+1)=\frac{1}{2}(X_{i_1(t)j}(t)+X_{i_2(t)j}(t))
       \end{align}
       if  $i\in\{i_1(t),i_2(t)\}$ and $j\in S(t)$. Otherwise, $X_{ij}(t+1)=X_{ij}(t)$. By a measurable set $S(t)$, we mean a set that is measurable with respect to the $\sigma$-algebra consisting of all subsets of $[n]$.
 \end{enumerate}
We refer to the preceding dynamical model as the \textit{Voting Diffusion Model}.
Equation~\eqref{eqn:dynrule} models the situation where agents $i_1(t),i_2(t)$ exchange their beliefs about the items in $S(t)$ only and move their opinions to the average of their current opinions. For the rest of the objects, the beliefs are not updated, and this is also true of the rest of the agents.

We assume that $\{S(t)\}$ is an arbitrary sequence of random sets which can depend on the past history (or even future) or some external disturbances (i.e.\ media news).
We refer to this process as the \textit{subject process}. The same holds for the communicating pair process $\{\{\iot,\itt\}\}$.
We refer to any process $\{X(t)\}$ that is generated using some communicating pair process and a subject process as a dynamics generated by a \textit{voting diffusion model}.

  \subsection{Some Examples of Voting Diffusion Process}\label{subsec:examples}
 Before analyzing the voting diffusion model, let us discuss some interesting instances of such process.
 \begin{itemize}
  \item \textit{Gossip Model}: The asynchronous gossip algorithm discussed in \cite{Boyd06} is the special case of the above dynamics with the following choice of processes:
      \begin{enumerate}
        \item The communicating pair process $\{\iot,\itt\}$ is an i.i.d.\ process.
        \item The number of candidates $n$ equals one, and hence, we naturally have $S(t)=\{1\}$ for all $t\geq 0$.
      \end{enumerate}
  \item \textit{Top-$k$ Selective Gossiping:} The top-$k$ selective gossiping model that was proposed and analyzed in \cite{RabbatCDC:2012} is a special case of the above dynamics with the following choice of the processes:
      \begin{enumerate}
        \item The communicating pair process $\{\{\iot,\itt\}\}$ is the same
        as in the gossip model of \cite{Boyd06}.
        \item Agent $i_1(t)$ and $i_2(t)$ discuss only the top-$k$ ranked objects.
        To describe it more precisely, for a vector $v\in \R^n$ and an integer $k\in [n]$,
        let $T_k(v)$ be the set of indices of $v$ corresponding to the top-$k$ positions in the ranking of entries of $v$, i.e.\ if $\sigma$ is a permutation on $[n]$ with $v_{\sigma(1)}\geq v_{\sigma(2)}\geq \cdots \geq v_{\sigma(n)}$, then
            \begin{align}\label{eqn:Tfuncdef}
              T_k(v)=\{j\in [n]\mid v_j\geq v_{\sigma(k)}\}.
            \end{align}
             Based on the definition of $T_k(\cdot)$, at time $t$, agents $i_1(t)$ and $i_2(t)$ only discuss voters in the set $S(t)=T_k(X_{i_1(t)}(t))\cup T_k(X_{i_1(t)}(t))$.
        This model is referred to as top-$k$ selective gossiping model.
      \end{enumerate}
  \item \textit{Binomial Selection:} In this case, we have an arbitrary communicating pair process $\{\{\iot,\itt\}\}$, while the subject process $\{S(t)\}$ is based on binomial object selection. Specifically, at time $t$,
   the set $S(t)\subseteq[n]$ consists of candidates that are obtained by choosing each candidate $j\in[n]$
   randomly and  independently with some probability $p\in(0,1]$.
  \item \textit{Hegselmann-Krause Gossiping:} On political issues, quite often, agents opinion are influenced by agents whose opinions are close to their own opinion. Motivated by the work in \cite{Krause:02}, we model such a dynamics as follows: let $i_1(t)$ and $i_2(t)$ be the agents that are chosen for vote updating at time $t$, then:
      \[S(t)=\{j\mid |X_{i_1(t)j}-X_{i_2(t)j}|\leq \epsilon\},\]
      where $\epsilon>0$ is a fixed parameter. In other words, at time $t$, agents $i_1(t)$ and $i_2(t)$ only discuss candidates for which their opinions are within $\e$-distance.
\end{itemize}
\section{Convergence Analysis}\label{sec:convergence}
In this section, we provide the convergence analysis of the voting diffusion dynamics. The main claim in this section is the following result.
\begin{theorem}\label{thrm:mainconvresult}
  For any communicating pair process $\{\{i_1(t),i_2(t)\}\}$ and a subject process $\{S(t)\}$, the dynamics $\{X(t)\}$ is convergent almost surely.
\end{theorem}

 To prove this result, for any candidate $j\in [n]$, let us define the $m\times m$ matrix process $\{\Wj(t)\}$ as follows:
{\small
 \begin{align}\label{eqn:Wjt}
   \Wj(t)= I-
   \bfo_{j\in S(t)}\frac{1}{2}(e_{i_1(t)}-e_{i_2(t)})(e_{i_1(t)}-e_{i_2(t)})^T.
 \end{align}}
It can be seen that if $S(t)$ and $\{i_1(t),i_2(t)\}$ are measurable, then $\Wj_{i\ell}(t)$ is a measurable random variable for any $i\in [m]$, $j\in [n]$ and $t\geq 0$.
Regardless whether $j\in S(t)$ or $j\not\in S(t)$, the matrix $\Wj(t)$ is (surely) doubly stochastic.
We use this property extensively in the following discussion.

From the definition of the process $\{\Wj(t)\}$ in~\eqref{eqn:Wjt} and the voting diffusion dynamics \eqref{eqn:dynrule}, it can be seen that for an arbitrary candidate $j\in [n]$,
  \begin{align}\label{eqn:jdynamics}
    \Xj(t+1)=\Wj(t)\Xj(t),\quad \mbox{for all $t\geq 0$},
  \end{align}
  where $\Xj(t)$ is the $j$th column of $X(t)$. From this, and the fact that the matrices $\Wj(t)$s are doubly stochastic almost surely,
  it immediately follows that the average opinion of the society on a particular candidate is preserved throughout the dynamics. In particular, we have
  \begin{align}\nonumber
 \frac{1}{m}\sum_{i=1}^mX_{ij}(t+1)&=\frac{1}{m}e^TX^{(j)}(t+1)\cr
  &=\frac{1}{m}e^T\Wj(t)\Xj(t)\cr
  &=\frac{1}{m}e^T\Xj(t)=\frac{1}{m}\sum_{i=1}^{m}X_{ij}(t).
  \end{align}

 In our subsequent derivations, we make extensive use of the following result, which is an immediate consequence of Theorem~5 in \cite{TouriNedich:Approx}.
 \begin{theorem}\label{TouriNedich:Approx}
   Let $\{A(t)\}$ be a sequence of $m\times m$ doubly stochastic matrices, and $A_{ii}(t)\geq \delta>0$ for all $t\geq 0$ and $i\in [m]$ and some $\delta>0$. Then, for any initial condition $x(0)\in \Rm$,
   the limit $\lim_{t\to\infty}x(t)$ exists. Furthermore, if we let the infinite flow graph of $\{A(t)\}$ be the graph $G=([m],E)$ with
   \[E=\{\{i_1,i_2\}\mid i_1,i_2\in [m], \sum_{t=0}^{\infty}A_{i_1i_2}(t)=\infty\},\]
   then for any $i_1,i_2$ belonging to the same connected component of $G$, we have
    \[\lim_{t\to\infty}(x_{i_1}(t)-x_{i_2}(t))=0.\]
 \end{theorem}
 \begin{proof}
   The result follows immediately by considering the trivial probability model, i.e.\ $\mathcal{F}=\{\emptyset,\Omega\}$ and the natural process $W(t,\omega)=A(t)$ for all $\omega\in \Omega$ and letting $\pi=\frac{1}{m}e$ in Theorem 5 in \cite{TouriNedich:Approx}.
 \end{proof}

  With this, proof of Theorem~\ref{thrm:mainconvresult} follows immediately.
  \begin{proof}(Proof of Theorem \ref{thrm:mainconvresult}) As discussed earlier,
  for any candidate $j\in [n]$, the sequence $\{\Wj(t)\}$ is almost surely doubly stochastic.
  Also, the dynamics governing the $j$th column of $X(t)$ is governed by equations \eqref{eqn:jdynamics}. Furthermore,
  $\Wj_{ii}(t)\geq \frac{1}{2}$, almost surely. Thus, by Theorem~\ref{TouriNedich:Approx}, $\lim_{t\to\infty}\Xj(t)$ exists almost surely.
  Since we have finitely many $j\in [n]$, it follows immediately that $\lim_{t\to\infty}X(t)$ exists.
  \end{proof}

By virtue of Theorem~\ref{thrm:mainconvresult}, we know that the limit
 $\lim_{t\to\infty}X(t)$ exists. We denote this limiting random matrix by
 $X(\infty)=\lim_{t\to\infty}X(t)$.

\subsection{Limiting Points of the Dynamics}
So far, we have proved that the aforementioned general information diffusion dynamics converges almost surely, independently of the choice of the communicating pair and the subject processes. Here, we study the convergent point of the dynamics. In particular, we are interested in determining agents which will eventually have the same opinion about a given candidate $j\in [n]$.
Before stating any result in this direction, let us define the concept of \textit{consensus} on a specific candidate.
   \begin{definition}\label{def:consensuspair}
     For a dynamics $\{X(t)\}$ generated by the voting diffusion model, for any two agents $i_1,i_2\in [m]$ and any candidate $j\in [m]$, we define the event $i_1\lra_j i_2$ as follows:
     \[i_1\lra_j i_2\triangleq\{\omega\in \Omega\mid \lim_{t\to\infty}(X_{i_1j}(t)-X_{i_2j}(t))=0\},\]
     or in other words, $i_1\lra_j i_2$ consists of the sample points over which agents $i_1$ and $i_2$ eventually consent on candidate $j$.
   \end{definition}

   Note that certain properties hold for events $i_1\lra_j i_2$. For example, for any triple $i_1,i_2,i_3\in [m]$, we have:
   \[(i_1\lra_j i_2) \cap (i_2\lra_j i_3)\subseteq (i_1\lra_j i_3),\]
   which follows immediately from the definition of these events.

In the upcoming discussion, we characterize the points in events $i_1\lra_j i_2$, based on the choice process and subject process for the diffusion model. For this, let us fix a choice process
$\{\{\iot,\itt\}\}$ and a subject process $\{S(t)\}$, and an initial profile $X(0)$. We associate with a candidate $j\in [n]$ a random graph $\Gj=([m],\Ej)$, where the edges of the graph specify the
set of agents which discuss item $j$ infinitely often (and of course, they themselves should appear in the communicating pair process infinitely often). More precisely,
   \begin{align}\label{eqn:Ejdef}
     \Ej&=\{\{i_1,i_2\}\mid i_1,i_2\in [m],\cr
     &\sum_{t=0}^{\infty}\bfo_{\{\{i_1(t),i_2(t)\}=\{i_1,i_2\} \} }
      \cdot \bfo_{\{j\in S(t)\} } =\infty\}.\quad
   \end{align}
Note that for any fixed $\{i_1,i_2\}$, the set $\{i_1,i_2\}\in \Ej$ is an event in our $\sigma$-field.
Thus, one can talk about the connectivity event. For this, let us define the event $i_1\lra_{\Gj}i_2$
as the event that agents $i_1$ and $i_2$ fall in the same connected component of $\Gj$.
   \begin{theorem}\label{thrm:infiniteflow}
     For any pair $\{i_1,i_2\}$ and any candidate $j\in[m]$, we have:
     \[i_1\lra_{\Gj}i_2\subseteq i_1\lra_j i_2.\]
   \end{theorem}
   \begin{proof}
 We claim that $\{i_1,i_2\}\in \Ej$ if and only if $\sum_{t=0}^{\infty}\Wj_{i_1i_2}(t)=\infty$.
 This can be seen by noting that $\Wj_{i_1i_2}(t)=\frac{1}{2}$ if and only if
$\{i_1,i_2\}=\{i_1(t),i_2(t)\}$ and $j\in S(t)$ at time $t$. Thus, for any sample point $\omega\in \Omega$, the graph $\Gj(\omega)$ is the infinite flow graph of the process $\{\Wj(t)\}$.
 Then, by Theorem \ref{TouriNedich:Approx}, we infer the claimed result.
   \end{proof}
 \section{Implications}\label{sec:implications}
  In this section, we discuss some implications of Theorem~\ref{thrm:mainconvresult} and Theorem~\ref{thrm:infiniteflow}.

Note that all the examples considered in Subsection~\ref{subsec:examples}
are examples of the voting diffusion model. Hence, by Theorem~\ref{thrm:mainconvresult}, it immediately follows that the dynamics generated by any of the given four models converges almost surely. In what follows, we present more detailed results on the extensions of the gossip model and the top-$k$ selective gossip model.

  \subsection{Gossip Model}
   Theorem~\ref{thrm:mainconvresult} asserts that convergence happens
   for any  i.i.d.\ choice of a communicating pairs process
   $\{\iot,\itt\}\}$. It also shows that the dynamics generated
   by the gossip model converges (almost surely) for an arbitrary communicating pair process $\{\{\iot,\itt\}\}$.  Specifically, Theorem~\ref{thrm:mainconvresult} shows that
   such a dynamics converges  even when $\{\{\iot,\itt\}\}$ is not i.i.d., as well as when it is not adapted to any filtering.

   Let us consider the case that the communicating pair process $\{\{\iot,\itt\}\}$ is adapted to a filtration $\{\F_t\}$ and $X(0)$ is measurable with respect to $\F_0$.
   We refer to this model as the adapted gossiping model. Then, it follows that $X(t)$ is measurable with respect to $\F_{t+1}$ and, hence, it is adapted to $\{\F_{t+1}\}$.

   Note that for the case of gossip model, we have only one object to discuss and, consequently, we can talk about the infinite flow graph of the model. Thus, in this case we may drop the superscripts $j$ for the infinite flow graph, as well as for the process $\{W(t)\}$.
   \begin{lemma}
     The infinite flow graph $G=([m],E)$ of the gossip model can be characterized as follows:
     \begin{align}\label{eqn:gossippattern}
       &E=\{\{i_1,i_2\}\mid \cr
     &\quad\sum_{t=0}^{\infty}\prob{\{i_1(t+1),i_2(t+1)\}=\{i_1,i_2\}\mid \F_t}=\infty\}.\quad
     \end{align}
   \end{lemma}
   \begin{proof}
     By the construction of the process $\{W(t)\}$, we have
     $\sum_{i}\sum_{j}W_{ij}(t)=\infty$, if and only if $\{\iot,\itt\}=\{i_1,i_2\}$ infinitely often.
     By the Borel-Cantelli lemma for conditional expectation (Theorem 5.3.2., \cite{Durrett}), it follows that $\{i_1,i_2\}\in E$ if and only if $\sum_{t=0}^{\infty}\prob{\{\iot,\itt\}=\{i_1,i_2\}\mid \F_t}=\infty$.
   \end{proof}

 By the preceding lemma and Theorem~\ref{thrm:infiniteflow}, for a given sample point, it suffices to consider the connectivity pattern
in the graph $([m],E)$ with edge set $E$ given in~\eqref{eqn:gossippattern}.
 Then, any two agents that fall in the same connected component of this graph
 will eventually consent on a common value.

 \subsection{Top-$k$ Selective Gossiping}
 In \cite{RabbatCDC:2012}, it is shown that for a given i.i.d.\ communicating pair process, the distance of the expected value of the dynamics $\{X(t)\}$ from its mean value is convergent,
 and the convergence property of the dynamics is left as an open problem.
 Theorem~\ref{thrm:mainconvresult} asserts a much more general result, as it shows that for an arbitrary communicating pair process $\{\{\iot,\itt\}\}$
 (not necessarily i.i.d.), the dynamics $\{X(t)\}$ is convergent almost surely.
 The point here is that Theorem~\ref{thrm:mainconvresult} makes no assumption on
 how the communicating pair process is constructed, whether it follows the standard gossip model or not,
 whether the underlying communication network is static or not.

Now, we discuss the top-$k$ selective gossiping in a more general setting
where $\{\{\iot,\itt\}\}$ is an arbitrary communicating pair process. We refer to it as the generalized top-$k$ selective gossiping model.
Throughout this section, we let the connectivity graph $G=([m],E)$ be
associated to a communicating pair process $\{\{\iot,\itt\}\}$, which is given by
  \begin{align}\label{eqn:infiniteoften}
  E=\{\{i,j\}\mid \{\iot,\itt\}=\{i_1,i_2\}\mbox{ i.o.}\},\end{align}
  where i.o.\ stands for infinitely often.

 As discussed in \cite{RabbatCDC:2012}, in the top-$k$ selective gossiping, the concern is
 whether the top-$k$ gossiping scheme can reach consensus on the top-$k$ candidates in the aggregate ranking of the initial opinion of the society, i.e., $\bar{X}(0)=\frac{1}{m}e^TX(0)$.
 In other words, the question is: do we have $i_1\lra_ji_2$ for any $i_1,i_2\in [m]$ and any $j\in T_k(\bar{X}(0))$? Our main shows that on the event set that $G$ is connected, there is a $k'$ such that agents consent on the top-$k'$ list using the generalized top-$k$ selective algorithm.

  Throughout the rest of our discussion, we work on a sample path of our dynamics for which $G$ is connected and hence, we assume that we have a deterministic sequence of $\{\iot,\itt\}$ that its connectivity graph is connected.

  For $\{i_1,i_2\}\in E$, let us define the following notations which will be useful
  in the upcoming development.
  \begin{itemize}
    \item We let $\beta_{i_1i_2}(s)$ be the $s$'th time instance when $\{i_1(t),i_2(t)\}=\{i_1,i_2\}$ occurs.
     Since $\{i_1,i_2\}\in E$, we have that $\{\beta_{i_1i_2}(s)\}$
     is an increasing sequence that goes to infinity.
    \item We let $S_{i_1i_2}^{\infty}$ be the set of candidates that appear infinitely often in
    the discussions between $i_1$ and $i_2$, i.e.,
    \begin{align}\nonumber
     &S_{i_1i_2}^{\infty}=\cap_{t\geq 0}\cup_{s\geq t}S(\beta_{i_1i_2}(s))\cr
      &=\cap_{t\geq 0}\cup_{s\geq t}(T_k(X_{i_1}(\beta_{i_1i_2}(s)))\cup T_k(X_{i_2}(\beta_{i_1i_2}(s)))).
    \end{align}
    \item We let $\al_{i_1i_2}=\min_{j\in S^\infty_{i_1i_2} }X_{i_1j}(\infty)$.
  \end{itemize}
  Note that $|S^{\infty}_{i_1i_2}|\geq k$ since each $S(\beta_{i_1i_2}(s))$ has cardinality at least $k$ and we have finitely many such subsets (at least one of these subsets should appear in
  the sequence  $\{S(\beta_{i_1i_2}(s))\}$ infinitely often).
  Also, note that $\al_{i_1i_2}$ may not appear to be well-defined but will subsequently be shown to hold true.

  The following lemma will assist us in proving the main result.
  \begin{lemma}\label{lemma:prelim}
    For $i_1,i_2\in [m]$ with $\{i_1,i_2\}\in E$, the following statements hold:
    \begin{enumerate}[(a)]
      \item \label{item:consentonlist} We have $X_{i_1j}(\infty)=X_{i_2j}(\infty)$ for all $j\in S_{i_1i_2}^{\infty}$.
      \item\label{item:part1} For any $j\not\in S_{i_1i_2}^{\infty}$ and any $i\in [m]$, we have
    \[X_{ij}(\infty)\leq \alpha_{i_1i_2}.\]
    In other words, eventually the opinion of all the agents on this candidate is less than $\al_{i_1i_2}$.

      \item\label{item:part2} For any other $i'_1,i'_2\in[m]$ with $\{i'_1,i'_2\}\in E$, we have $\al_{i_1i_2}=\al_{i'_1i'_2}$, i.e.,\ the minimum of the list is independent of the choice of the communicating pair.
      \item\label{item:part3} If for some $j\in [n]$ and some $i\in [m]$, we have $X_{ij}(\infty)>\al_{i_1i_2}$, then $j\in S_{i'_1i'_2}$ for all $\{i'_1,i'_2\}\in E$.
    \end{enumerate}
  \end{lemma}
  \begin{proof}
    \eqref{item:consentonlist} \ Note that if $j\in S_{i_1i_2}^{\infty}$ it follows that $i_1$ and $i_2$
    talk about $j$ infinitely often. Then, by Theorem~\ref{thrm:infiniteflow} we have
    $X_{i_1j}(\infty)=X_{i_2j}(\infty)$, which implies that $\al_{i_1i_2}$ is well-defined. \\
    \eqref{item:part1} \ We first show that the claim holds for agents $i_1,i_2$. Let $\ell\in  S_{i_1i_2}^{\infty}$. Since $j\not\in S_{i_1i_2}^{\infty}$, it follows that after some $t_0\geq 0$, we have $j\not\in S(\beta_{i_1i_2}(t_0))$.
    Thus, at time instance $\beta_{i_1i_2}(s)$, $s\ge t_0$, when $i_1$ and $i_2$ talk about $\ell$ we should have $X_{i_1\ell}(s)\geq X_{i_1j}(s)$. Since $\lim_{t\to\infty}X(t)$ exists and also $X_{i_1\ell}(s)\geq X_{i_1j}(s)$ for infinitely many $s$, it follows that
    \begin{align}\nonumber
   X_{i_1\ell}(\infty) =  \lim_{t\to\infty}X_{i_1\ell}(t)
    \ge \lim_{t\to\infty}X_{i_1j}(t) = X_{i_1j}(\infty).
    \end{align}
    Now, consider an arbitrary neighbor $i$ of $i_1$ in $G$.
    Then, if $X_{ij}(\infty)>\alpha_{i_1i_2}$, and since $X_{i_1j}(\infty)\leq \alpha_{i_1i_2}$ then by Theorem~\ref{thrm:infiniteflow}, it follows that there is some time $t\geq 0$ such that $i_1$ and $i$ do not talk about $j$ after time $t$.
    Otherwise, by Theorem~\ref{thrm:infiniteflow}, it follows that $X_{i_1j}(\infty)=X_{ij}(\infty)$ which is contradiction. But this means that at any time instance such as $s\ge t$ that $i_1$ and $i$ talk, there exists some subset $S(s)$ of cardinality at least $k$ such that $X_{i \ell}(s)\ge X_{i j}(\infty)>\alpha_{i_1i_2}$. But this implies that there is a set of cardinality at least $k$, such that
    \[X_{i_1 \ell}(\infty)=X_{i\ell}(\infty)\ge X_{ij}(\infty)>\alpha_{i_1i_2}.\]
     This itself implies that there is a set $S$ of cardinality at least $k$ such that $i_1,i_2$ will talk about
   the items $\ell \in S$ infinitely often and
     \[X_{i_1\ell}(\infty)=X_{i_2\ell}(\infty)\ge \beta>\alpha_{i_1i_2}\quad\hbox{for all $\ell\in S$ },\]
     which contradicts with the fact that $\alpha_{i_1i_2}=\min_{\ell\in S^{\infty}}X_{i_1\ell}(\infty)$.

     Using a similar line of argument, we can show that for any neighbor $\gamma$ of $i$ in $G$, we have $X_{\gamma j}(\infty)\le \al_{i_1i_2}$. Since the graph $G$ is connected, we have $X_{ij}(\infty)\leq \al_{i_1i_2}$ for all $i\in [m]$.\\
    \eqref{item:part2} \  Let $i$ be an arbitrary neighbor of $i_2$ in $G$, other than $i_1$.  Then by Theorem~\ref{thrm:infiniteflow}, for $j\in S^{\infty}_{i_1i_2}\cap S^{\infty}_{i_2i}$, we have $X_{i_1j}(\infty)=X_{i_2j}(\infty)=X_{ij}(\infty)$ and for $j\in S^{\infty}_{i_1i_2}\setminus S^{\infty}_{i_2i}$ we have $X_{i_1j}(\infty)\leq \al_{i_2i}$. Thus, $\al_{i_1i_2}\leq \al_{i_2i}$. Using a similar argument we have $\al_{i_1i_2}\geq \al_{i_2i}$ and, hence $\al_{i_1i_2}=\al_{i_2i}$. Since the graph $G$ is connected, it follows that $\al_{i_1i_2}=\al_{i'_1i'_2}$ for all $\{i'_1,i'_2\}\in E$. \\
    \eqref{item:part3} \  This follows immediately from \eqref{item:part1} and \eqref{item:part2}.
  \end{proof}

  Based on Lemma~\ref{lemma:prelim}, we can prove our main claim
  that the agents consent on the top-$k$ aggregate list.
  \begin{theorem}\label{thrm:topk}
    Let $G=([m],E)$ be the graph with the edge set as given in~\eqref{eqn:infiniteoften}.
  Let $\{X(t)\}$ be a dynamics generated by the generalized top-$k$ selective gossiping model. Then, for any $\omega\in \{\mbox{$G$ is connected }\}$, there exists an $k'(\omega)\geq 1$ such that the society consents on the top-$k'$ aggregate ranking,
  i.e.\ $i_1\lra_j i_2$ for any $j\in T_{k'}(\bar{X}(0))$, where $\bar{X}(0)=\frac{1}{m}e^TX(0)$.
  \end{theorem}
  \begin{proof}
   Fix a sample point $\omega\in \{\mbox{$G$ is connected}\}$. Let $\al=\al_{i_1i_2}(\omega)$ for some $\{i_1,i_2\}\in E$ and $Q=\{j\in[n]\mid \bar{X}_{j}(0)\geq\al\}$. We first prove that for any candidate $j\in Q$, we have consensus in society on $j$, i.e.\ $i_1\lra_j i_2$ for any $i_1,i_2\in[m]$.

   Note that if for some $j\in [n]$, $\bar{X}_j(0)>\al$, then since the average of $X(t)$ is preserved throughout the time, we should have $X_{ij}(\infty)>\al$ for some $i\in [m]$. By Lemma~\ref{lemma:prelim}-\eqref{item:part3} it follows that $j\in \cap_{\{i_1,i_2\}\in E}S^{\infty}_{i_1i_2}$. Thus, $i_1\lra_j i_2$ for all $i_1,i_2\in[m]$ and hence, the society consents on item $j$.

   If for some $j\in [n]$, $\bar{X}_j(0)=\al$ and agents do not consent on item $j$, then since the average is preserved throughout the dynamics, it follows that there exists some $i\in [m]$ such that $X_{ij}(\infty)>\al$. By a similar argument as in the previous case, this implies that the society consent on item $j$.

   Thus, it follows that the society consent on the top-$|Q|$ aggregate ranking of the society and hence, the result follows immediately.
  \end{proof}

 \section{Conclusion}\label{sec:conclusion}
 In this work we presented a general dynamics for voters' opinion dynamics in a time-varying random network and proved the convergence of such dynamics. Based on the proposed diffusion model, we provided generalizations to the gossip model and top-$k$ selective gossiping model, and discussed the convergence implications for these models. Many questions are left to be answered. Among them, the convergence rate analysis of these dynamics for specific selection processes. For example, the study of binomial selection process introduced in Section~\ref{sec:generaldynamics} is of particular interest.
\bibliographystyle{plain}
\bibliography{distributedagg}
\end{document}